\documentclass[11pt]{article}

\usepackage[a4paper,margin=1in]{geometry}
\usepackage[T1]{fontenc}
\usepackage[utf8]{inputenc}
\usepackage{lmodern}

\usepackage{amsmath,amssymb,amsthm,mathtools}
\usepackage{mathrsfs}
\usepackage{xcolor}
\usepackage{hyperref}

\hypersetup{
  colorlinks=true,
  linkcolor=blue,
  citecolor=blue,
  urlcolor=blue
}

\numberwithin{equation}{section}

\newtheorem{theorem}{Theorem}[section]
\newtheorem{definition}[theorem]{Definition}
\newtheorem{proposition}[theorem]{Proposition}
\newtheorem{corollary}[theorem]{Corollary}
\newtheorem{lemma}[theorem]{Lemma}
\newtheorem{remark}[theorem]{Remark}

\newtheorem{assumption}[theorem]{Assumption}

\title{Volume-Independent Spectral Stability of Energy-Truncated Effective Hamiltonians in Quantum Spin Systems}

\author{
Ayumi Ukai$^{1,2}$\\
{\small $^1$ RIKEN Center for Quantum Computing, Wako, Japan}\\
{\small $^2$ Research Institute for Mathematical Sciences, Kyoto University, Kyoto, Japan}\\
{\small \texttt{ayumi.ukai@riken.jp}}
}

\begin{document}

\maketitle
\tableofcontents

\begin{abstract}
We prove a volume-uniform effective-Hamiltonian theorem for bounded finite-range quantum spin systems on possibly infinite lattices. For any finite target region, we construct an energy-truncated Hamiltonian and prove a volume-uniform spectral-overlap bound controlling the leakage of its low-energy spectral subspace into the high-energy spectral subspace of the original Hamiltonian. The bound may contain non-exponential spectral-window terms, but its cutoff-dependent remainder decays exponentially in the cutoff.

In finite volume, this yields stability of low-lying eigenvalues, with eigenvalue errors controlled by the exponentially small cutoff-dependent remainder. In infinite volume, we prove the corresponding spectral-overlap estimate in the GNS representation of an infinite-volume ground state. Thus, for bounded finite-range interactions, we extend and strengthen the effective-Hamiltonian mechanism of \cite{Arad-Kuwahara-Landau 2016} by replacing the finite-volume operator-norm formulation with a volume-uniform spectral-overlap formulation applicable in the thermodynamic limit.
\end{abstract}

\section{Introduction}
Effective Hamiltonians obtained by truncating the high-energy part of a Hamiltonian are a basic tool in the analysis of low-energy structures of quantum many-body systems. In many arguments, one would like to replace a Hamiltonian by a bounded effective Hamiltonian while preserving its low-energy spectrum with high accuracy. 

A foundational result in this direction was obtained by Arad, Kuwahara, and Landau \cite{Arad-Kuwahara-Landau 2016}, who proved a finite-volume effective-Hamiltonian estimate for quantum spin systems. This estimate has played an essential role in subsequent area-law arguments, including results for polynomially degenerate gapped ground spaces and for states with low energy density \cite{Arad-Landau-Vazirani-Vidick2017}.

A second motivation comes from infinite systems.  Many structural questions in quantum many-body theory are naturally formulated in the thermodynamic limit, for example in the study of phase classification and symmetry-protected topological phases \cite{Ogata2022}. It is therefore important to understand whether the effective-Hamiltonian construction can be made uniform in the system size and applied directly to infinite-volume systems.

These two motivations lead us to revisit the effective-Hamiltonian construction of \cite{Arad-Kuwahara-Landau 2016} from the viewpoint of volume-uniform spectral stability.

\subsection{A volume-dependent obstruction}
There is a structural obstruction to extending the finite-volume operator-norm formulation of \cite{Arad-Kuwahara-Landau 2016} directly as a volume-uniform statement. The obstruction comes from trying to prove a volume-uniform operator-norm bound of the form
\[
\|(H-\bar{H})\bar{P}_{\mathrm{low}}\|,
\]
where $\bar{P}_{\mathrm{low}}$ is the low-energy spectral projection of the energy-truncated Hamiltonian $\bar{H}$.

As we show in Appendix \ref{S_counterexample}, the corresponding AKL-type operator-norm estimate, displayed in \eqref{eq_AKL_original}, necessarily exhibits volume dependence in general. The estimate of \cite{Arad-Kuwahara-Landau 2016} is valid for each fixed finite volume, but its claimed volume-uniformity does not hold in general. The example in Appendix \ref{S_counterexample} shows that the corresponding operator-norm bound can necessarily depend on the total volume.  This is the reason why we replace the operator-norm formulation by a spectral-overlap formulation.

We therefore replace direct operator-norm control by a spectral-overlap estimate. This estimate controls the relative position of low-energy spectral subspaces directly, and is strong enough to imply finite-volume eigenvalue stability while remaining meaningful for infinite-volume Hamiltonians, where individual low-energy eigenvalues need not exist. The same volume dependence is already reflected in related estimates such as \cite[equation (283)]{Kuwahara-Saito2020}.

Thus the difficulty is not merely a matter of improving constants in the finite-volume proof.  Rather, the operator-norm formulation itself is too strong for a volume-uniform theory. The stability statement itself must therefore be reformulated.

\subsection{Spectral-overlap method and main results}
The reformulation used in this paper is based on spectral overlap. Instead of estimating $H-\bar{H}$ on a low-energy subspace, we estimate how much the low-energy spectral subspace of $\bar{H}$ leaks into the high-energy spectral subspace of $H$. The central quantity is
\[
\left\|\bigl(I-E^H(-\infty,q)\bigr)
E^{\bar{H}}(-\infty,p]\right\|.
\]

Our main theorem gives a volume-uniform bound on this spectral leakage. The cutoff-dependent part of the bound is exponentially small in the cutoff $M$, with constants depending only on local interaction parameters and the boundary region, and not on the total system size.

In finite volume, this spectral-overlap estimate is Theorem \ref{thm_main_finite}. It gives the desired effective-Hamiltonian stability statement in the spectral-overlap sense, with constants independent of the total volume, and does not require a volume-uniform operator-norm bound for $\|(H-\bar{H})\bar{P}_{\mathrm{low}}\|$; the resulting eigenvalue comparison is Theorem \ref{thm_main2_finite}, which can be read informally as follows.

\paragraph{Finite-volume reading of the main estimate.}
For a finite-volume energy-truncated Hamiltonian $\bar{H}$ with cutoff $M$, Theorem \ref{thm_main2_finite} can be read informally as saying that the low-lying eigenvalues of $\bar{H}$ satisfy
\[
\epsilon_j-e^{-c(M-2\epsilon_j)}\leq \bar\epsilon_j\leq\epsilon_j,
\]
where $\epsilon_j$ and $\bar{\epsilon}_j$ denote the corresponding low-lying eigenvalues of $H$ and $\bar{H}$, respectively. The constant $c>0$ in the error term is independent of the total volume.
\bigskip

The infinite-volume main result is Theorem \ref{thm_main_infinite}: it extends the spectral-overlap estimate itself, rather than the finite-volume eigenvalue formulation, to the GNS representation of an infinite-volume ground state. In this sense, it gives a thermodynamic-limit, volume-uniform spectral-overlap version of the effective-Hamiltonian construction of \cite{Arad-Kuwahara-Landau 2016}. A rank-threshold consequence and the gapped-ground-state consequence are stated separately in Corollary \ref{cor_main2_infinite}; the rank-threshold statement becomes an ordinary eigenvalue comparison when the relevant low-energy spectrum is discrete, isolated, and of finite multiplicity.

\subsection{Relation to previous finite-volume work}
The present paper extends and strengthens the effective-Hamiltonian construction of Arad, Kuwahara, and Landau \cite{Arad-Kuwahara-Landau 2016} by reformulating it as a volume-uniform spectral-overlap estimate. The resulting formulation is uniform in the total system size and yields eigenvalue errors that are exponentially small in the cutoff $M$.

The point of the reformulation is that the spectral-overlap estimate is tailored to the spectral conclusions needed here and remains uniform in the total volume, whereas the corresponding operator-norm formulation is not volume-uniform in general. Throughout the paper we work with bounded finite-range interactions. This assumption allows the imaginary-time estimates to be controlled by local interaction parameters and the boundary region, rather than by the total system size; see Remark \ref{rmk_range}.

\subsection{Relation to previous infinite-volume work}

The main infinite-volume result of this paper is a spectral-overlap estimate for general low-energy spectral windows in the GNS representation. This is broader than a stability statement for the ground-state projection: the gapped locally unique ground-state case is recovered as a consequence in Corollary \ref{cor_main2_infinite}.

For the gapped ground-state sector, Kuwahara and Saito \cite{Kuwahara-Saito2020} developed a stability mechanism, and Liu, Yi, Zhou, and Zou \cite{Liu-Yi-Zhou-Zou2025} formulated an infinite-volume version covering local uniqueness and spectral-gap stability. The present paper is complementary to this line of work: it treats general low-energy spectral windows through a spectral-overlap estimate and gives a self-contained justification of the analytic and domain issues arising from unbounded GNS Hamiltonians; see Proposition \ref{prop_AKL} and Remark \ref{rmk_hadamard_infinite}.

This broader formulation may also be useful beyond ground-state stability, for example in future studies of low-energy approximations of quantum dynamics, effective-Hamiltonian approximations of thermal equilibrium states, and rigorous questions related to ensemble equivalence and the eigenstate thermalization hypothesis.

\subsection{Organization of the paper}
Section \ref{S_notation} introduces the notation and the general setup for quantum spin systems. It also recalls a basic lemma on the overlap of spectral projections, which is the main tool replacing direct operator-norm perturbation estimates. Section \ref{S_finite} treats finite-volume systems. We define the energy-truncated Hamiltonian, prove the finite-volume spectral-overlap estimate, and derive the stability of low-lying eigenvalues. Section \ref{S_infinite} extends the construction to infinite systems. Using the GNS representation of an infinite-volume ground state, we prove the infinite volume version of the effective Hamiltonian theorem.

\subsection*{Acknowledgements}
The author thanks Tomotaka Kuwahara for fruitful discussions and helpful comments on the manuscript. The author is also grateful to Itai Arad for his kind hospitality during the author's stay at the National University of Singapore. The author thanks Yoshiko Ogata for valuable guidance. This work was supported by the RIKEN Junior Research Associate Program. The author also acknowledges the Hakubi Project of RIKEN.

\section{Notation and general setup}\label{S_notation}

This section introduces the common setup used for both finite and infinite systems. In particular, Section \ref{S_finite} relies only on finite-volume objects defined here, while infinite-volume notions are introduced in Section \ref{S_infinite}.

We consider a quantum spin system on a possibly infinite $\nu$-dimensional lattice $\Gamma\subset\mathbb{Z}^\nu$, equipped with the $\ell^1$-distance $d(x,y):=\sum_{j=1}^\nu |x_j-y_j|$. At each site $x\in\Gamma$, we place a finite-dimensional quantum spin (of dimension $d$). For each finite region $\Lambda\Subset\Gamma$, we denote by $\mathcal{A}_\Lambda$ the algebra of observables supported on $\Lambda$, namely, 
\[
\mathcal{A}_\Lambda
:=\bigotimes_{x\in \Lambda} M_d(\mathbb{C})
=\mathcal{B}(\mathcal{H}_\Lambda),\quad
\mathcal{H}_\Lambda:=\bigotimes_{x\in \Lambda}\mathbb{C}^d,
\]
where $\mathcal{B}(\mathcal{H}_\Lambda)$ denotes the full matrix algebra on the finite-dimensional Hilbert space $\mathcal{H}_\Lambda$. Here and below, $X\Subset\Gamma$ means that $X$ is a finite subset of $\Gamma$.

Interactions among the spins are described by a map $\Phi:X\Subset\Gamma\mapsto\Phi(X)=\Phi(X)^*\in\mathcal{A}_X$. For a finite-volume system $\Lambda\Subset\Gamma$, the Hamiltonian and the Heisenberg dynamics are defined by
\[
H_\Lambda:=\sum_{X\subset\Lambda} \Phi(X),\qquad
\alpha_t^\Lambda(A):=e^{itH_\Lambda}Ae^{-itH_\Lambda}.
\]
\begin{assumption}\label{ass_interaction}
In this paper, we assume that the interaction $\Phi$ is bounded and of finite range:
\begin{itemize}
\item (finite range)
\(
\Phi(X)=0 \ \text{whenever}\ 
\operatorname{diam}(X):=\max_{x,y\in X} d(x,y)>\mathfrak{r}_\Phi.
\)
\item (bounded)
\(
\mathfrak{j}_\Phi:=\sup_{x\in\Gamma}
\sum_{\substack{X\Subset\Gamma\\ x\in X}} \|\Phi(X)\|<\infty.
\)
\end{itemize}
\end{assumption}
We do not assume translation invariance.

\begin{remark}\label{rmk_range}
Let us briefly comment on the assumptions. The framework of \cite{Arad-Kuwahara-Landau 2016} allows $k$-local interactions which may be long range, whereas the present paper assumes bounded finite-range interactions.

Under this assumption, the interaction is uniformly $N_\Phi$-local, in the sense that every nonzero interaction term $\Phi(X)$ satisfies $|X|\leq N_\Phi$, where
\[
N_\Phi:=\sup_{x\in\Gamma}
\left|\{y\in\Gamma\,;\, d(x,y)\leq \mathfrak{r}_\Phi\}\right|.
\]

This finite-range locality is used to obtain imaginary-time estimates with constants depending only on local interaction parameters and the boundary region, rather than on the total volume. Extending the corresponding estimates to long-range interactions would require additional care; see Remark \ref{rmk_hadamard_finite}.
\end{remark}

\subsection{A preliminary lemma on spectral overlap}
A key technical ingredient of this paper is a quantitative control of the overlap between spectral subspaces of the original Hamiltonian and a truncated Hamiltonian. In this subsection, we introduce a basic tool to quantify the overlap between two projections.

\begin{lemma}[{\cite[Lemma 2]{Arad-Landau-Vazirani-Vidick2017}}]\label{lem_overlap}
Let $\mathcal{H}$ be an arbitrary Hilbert space, and let $P,Q\in\mathcal{B}(\mathcal{H})$ be projections. If $\|(I-Q)P\|\leq c$ for some $c\in(0,1)$, then for any $\psi\in P\mathcal{H}$, we have
\[
\operatorname{Rank}(P)\leq \operatorname{Rank}(Q),\quad \|Q\psi\|\geq \sqrt{1-c^2}\,\|\psi\|.
\]
\end{lemma}

\begin{proof}
Since $\|(I-Q)P\|\leq c$, we have
\[
P(I-Q)P=((I-Q)P)^*((I-Q)P)\leq c^2P.
\]
Therefore,
\[
PQP=P-P(I-Q)P\geq(1-c^2)P.
\]
Hence, for any $\psi\in P\mathcal{H}$, we obtain
\[
\|Q\psi\|^2=\langle\psi,Q\psi\rangle =\langle\psi,PQP\psi\rangle\geq (1-c^2)\|\psi\|^2.
\]
This gives
\[
\|Q\psi\|\geq\sqrt{1-c^2}\,\|\psi\|.
\]

In particular, if $\psi\in P\mathcal{H}$ and $Q\psi=0$, then the above inequality implies
\[
0=\|Q\psi\|\geq\sqrt{1-c^2}\,\|\psi\|.
\]
Since $c<1$, we have $\sqrt{1-c^2}>0$, and hence $\psi=0$. Thus $Q$ is injective on the range of $P$. Consequently,
\[
\operatorname{Rank}(P)\leq\operatorname{Rank}(Q).
\]
\end{proof}

\section{Finite Systems}\label{S_finite}
\subsection{Setup}
In this section, we assume without loss of generality that the lowest eigenvalue $\epsilon_0=0$, and especially
\[
H_\Lambda\geq 0.
\]

For each $L\subset\Lambda$, we decompose the Hamiltonian into three parts:
\begin{equation}\label{eq_decompose}
H_\Lambda=H'_L+H_{\partial L}+H'_{L^c},
\end{equation}
where $\partial L$ is a region containing the interaction boundary between $L$ and $L^c$, namely,
\[
\partial L\supset 
\{x\in L\,;\,d(x,L^c)\leq\mathfrak{r}_\Phi\}\cup
\{y\in L^c\,;\,d(y,L)\leq\mathfrak{r}_\Phi\},
\]
and
\[
H'_L:=\sum_{\substack{Z\subset L\\ Z\cap(\partial L)^c\neq\emptyset}}\Phi(Z),\quad
H_{\partial L}:=\sum_{Z\subset\partial L}\Phi(Z),\quad
H'_{L^c}:=\sum_{\substack{Z\subset L^c\\ Z\cap(\partial L)^c\neq\emptyset}}\Phi(Z).
\]

We denote by $E^{H_\Lambda}(I)$ the spectral projection of $H_\Lambda$ corresponding to the energy interval $I\subset\mathbb{R}$. In particular, if
\[
0=\epsilon_0\leq\epsilon_1\leq\cdots\leq\epsilon_{d^{|\Lambda|}-1}
\]
are the eigenvalues of $H_\Lambda$, and if
\[
\psi_0,\psi_1,\cdots,\psi_{d^{|\Lambda|}-1}
\]
are corresponding orthonormal eigenvectors with eigenvalue $\epsilon_j$, then for any $a\leq b$,
\[
E^{H_\Lambda}([a,b]):=
\sum_{\substack{j=0,1,2,\ldots,d^{|\Lambda|}-1\\ a\leq\epsilon_j\leq b}} |\psi_j\rangle\langle\psi_j|.
\]

\begin{definition}[Energy-truncated Hamiltonian]
For any $M>\|H_{\partial L}\|$, we define
\begin{align*}
\overline{H'_L+H'_{L^c}}
&:=\big(H'_L+H'_{L^c}\big)\,E^{H'_L+H'_{L^c}}(-M,M)
+ M\,E^{H'_L+H'_{L^c}}[M,\infty),\\
\bar{H}_\Lambda
&:=H_{\partial L}+ \overline{H'_L+H'_{L^c}},
\end{align*}
where $E^{H'_L+H'_{L^c}}(I)$ denotes the spectral projection of $H'_L+H'_{L^c}$ corresponding to the interval $I$.
\end{definition}
Here, we note that
\[
H'_L+H'_{L^c}=H_\Lambda-H_{\partial L}\geq -\|H_{\partial L}\|
\]
and
\[
\bar{H}_\Lambda\leq H_\Lambda \quad\text{and}\quad
\|\bar{H}_\Lambda\|\leq\|H_{\partial L}\|+M.
\]

\subsection{Main Theorem}
We denote by $E^{\bar{H}_\Lambda}(I)$, $\{\bar{\epsilon}_j\}$ and $\{\bar{\psi}_j\}$ the spectral projection, the eigenvalues and the corresponding orthonormal eigenvectors of $\bar{H}_\Lambda$, respectively.

\begin{theorem}\label{thm_main_finite}
We set
\[
\lambda:=\frac{1}{4\mathfrak{j}_\Phi(N_\Phi+|\partial L|)},\quad N_\Phi:=\sup_{x\in\Gamma} |\{ y\in\Gamma \,;\, d(x,y)\leq \mathfrak{r}_\Phi\}|<\infty,
\]
\begin{align*}
\delta(p,q)&:=2\sqrt{2}(M+5\|H_{\partial L}\|+q) \exp(-\lambda(M-2p-18\|H_{\partial L}\|)),\\
\eta(\epsilon,\delta)&:=2\sqrt{2}(M+\|H_{\partial L}\|+\delta) \exp(-\lambda(M-2\epsilon-10\|H_{\partial L}\|)).
\end{align*}
\begin{itemize}
\item[(i)] For any $q>p\geq-2\|H_{\partial L}\|$, we have
\[
\left\|\Big(I-E^{H_\Lambda}(-\infty,q)\Big)\,
E^{\bar{H}_\Lambda}(-\infty,p]\right\|
\leq\frac{p+2\|H_{\partial L}\|+\delta(p,q)}{q+2\|H_{\partial L}\|}.
\]
\item[(ii)] For any $\delta>\epsilon\geq0$, we have
\[
\left\|\Big(I-E^{H_\Lambda}(-\delta,\delta)\Big)\,
E^{\bar{H}_\Lambda}[-\epsilon,\epsilon]\right\|
\leq\frac{\epsilon+\eta(\epsilon,\delta)}{\delta}.
\]
\end{itemize}
\end{theorem}

Here and below, ``volume-uniform'' means uniform in the total volume $|\Lambda|$.  The constants may still depend on the local interaction parameters and on the boundary region $\partial L$.

\begin{remark}
Theorem \ref{thm_main_finite} is weaker than the operator-norm estimate
\[
\left\|\bigl(H_\Lambda-\bar H_\Lambda\bigr)
E^{\bar H_\Lambda}(-\infty,\epsilon]\right\|,
\]
of the type appearing in \cite[Theorem 2.6]{Arad-Kuwahara-Landau 2016}. 
Instead, it controls only the spectral leakage of low-energy vectors of $\bar{H}_\Lambda$ into high-energy subspaces of $H_\Lambda$.

This spectral-overlap formulation is used because it is volume-uniform: as Appendix \ref{S_counterexample} shows, the corresponding operator-norm estimate can grow with the total volume, whereas the constants in Theorem \ref{thm_main_finite} are independent of $|\Lambda|$. Nevertheless, the spectral-overlap estimate is sufficient for the eigenvalue comparison in Theorem \ref{thm_main2_finite}.
\end{remark}

\begin{theorem}\label{thm_main2_finite}
For any $j$ with $0\leq j\leq d^{|\Lambda|}-1$, we set
\begin{align*}
\delta_j&:=2\sqrt{2}(M+5\|H_{\partial L}\|+\epsilon_j) \exp(-\lambda(M-2\epsilon_j-18\|H_{\partial L}\|))\\
\eta&:=\eta(2\delta_0,\Delta)=2\sqrt{2}(M+\|H_{\partial L}\|+\Delta)\exp(-\lambda(M-4\delta_0-10\|H_{\partial L}\|)).
\end{align*}
\begin{itemize}
\item[(i)] We have
\[
\epsilon_j-2\delta_j\leq \bar{\epsilon}_j\leq \epsilon_j.
\]
\item[(ii)] If $\epsilon_1=\Delta>2\delta_0+\eta$, we have
\[
|\langle\psi_0,\bar{\psi}_0\rangle|^2\geq 1-\left(\frac{2\delta_0+\eta}{\Delta}\right)^2.
\]
\end{itemize}
\end{theorem}
\begin{proof}
(i): First, $\bar{\epsilon}_j\leq\epsilon_j$ follows from $\bar{H}_\Lambda\leq H_\Lambda$. It remains to prove the lower bound $\epsilon_j-2\delta_j\leq\bar{\epsilon}_j$. If $\epsilon_j-2\delta_j<-2\|H_{\partial L}\|$, then $\bar{H}_\Lambda\geq-2\|H_{\partial L}\|$ implies $\bar{\epsilon}_j\geq-2\|H_{\partial L}\|>\epsilon_j-2\delta_j$. 

Hence, it is enough to consider the case $\epsilon_j-2\delta_j\geq -2\|H_{\partial L}\|$. In this case, for any $j=0,1,2,\ldots$, since $\delta(\epsilon_j-2\delta_j,\epsilon_j)\leq \delta_j$ and by Theorem \ref{thm_main_finite}(i) with $(p,q)=(\epsilon_j-2\delta_j,\epsilon_j)$, we have
\begin{align*}
&\left\|\left(I-E^{H_\Lambda}(-\infty,\epsilon_j)\right)\,
E^{\bar{H}_\Lambda}(-\infty,\epsilon_j-2\delta_j]\right\|\\
&\leq\frac{\epsilon_j+2\|H_{\partial L}\|-2\delta_j+\delta(\epsilon_j-2\delta_j,\epsilon_j)}{\epsilon_j+2\|H_{\partial L}\|}
\leq 1-\frac{\delta_j}{\epsilon_j+2\|H_{\partial L}\|}<1.
\end{align*}
By Lemma \ref{lem_overlap},
\[
\operatorname{Rank}(E^{\bar{H}_\Lambda}(-\infty,\epsilon_j-2\delta_j])\leq \operatorname{Rank}(E^{H_\Lambda}(-\infty,\epsilon_j))<j+1.
\]
Since
\[
\bar{\epsilon}_j:=\inf\left\{p\in\mathbb{R} \,|\, \operatorname{Rank} \left(E^{\bar{H}_\Lambda}(-\infty,p]\right)\geq j+1\right\},
\]
we have $\epsilon_j-2\delta_j\leq\bar{\epsilon}_j$.

(ii): By Theorem \ref{thm_main_finite}(ii) with $(\epsilon,\delta)=(2\delta_0,\Delta)$, we have
\[
\left\|\left(I-E^{H_\Lambda}(-\Delta,\Delta)\right)\,
E^{\bar{H}_\Lambda}[-2\delta_0,2\delta_0]\right\|
\leq \frac{2\delta_0+\eta(2\delta_0,\Delta)}{\Delta}.
\]
By part (i) with $j=0$, we have $-2\delta_0\leq \bar{\epsilon}_0\leq 0$, and hence 
\[
\bar{\psi}_0\in E^{\bar{H}_\Lambda}[-2\delta_0,2\delta_0]\mathcal H_\Lambda.
\]
Since $H_\Lambda\geq0$ and $\epsilon_1=\Delta$, we have
\(E^{H_\Lambda}(-\Delta,\Delta)=|\psi_0\rangle\langle\psi_0|\).
Hence, by Lemma \ref{lem_overlap},
\[
|\langle\psi_0,\bar{\psi}_0\rangle|^2
=\|E^{H_\Lambda}(-\Delta,\Delta) \bar{\psi}_0\|^2
\geq 1-
\left(\frac{2\delta_0+\eta(2\delta_0,\Delta)}{\Delta}\right)^2.
\]
\end{proof}

\subsection{Lemmas and Proof of Theorem \ref{thm_main_finite}}
Theorem \ref{thm_main_finite} is deduced from Lemma \ref{lem_AKL2.2_6.1} below, which corresponds to \cite[Theorem 2.2 and Theorem 6.1]{Arad-Kuwahara-Landau 2016}. Its proof relies on two off-diagonal estimates: namely Lemmas \ref{lem_AKL2.1_5.1} and \ref{lem_AKL6.2}.

Lemma \ref{lem_AKL2.1_5.1} revisits \cite[Theorems 2.1 and 5.1]{Arad-Kuwahara-Landau 2016}. We include a proof adapted to the present volume-uniform formulation, making explicit a locality point that is useful for our later infinite-volume argument. Lemma \ref{lem_AKL6.2} corresponds to \cite[Lemma 6.2]{Arad-Kuwahara-Landau 2016}; its proof follows the same strategy as in the previous work.

\begin{lemma}\label{lem_AKL2.1_5.1}
We set $\lambda$ as in Theorem \ref{thm_main_finite}. Then, for any $M,N\in\mathbb{R}$ and any $A\in\mathcal{A}_\Lambda$, we have
\begin{align*}
&\|E^{H_\Lambda}[M,\infty) \,A\, E^{H_\Lambda}(-\infty,N]\|\\
&\quad\leq \exp(-\lambda(M-N-4\|H_{\partial L}\|)) \left\|e^{\lambda(H'_L+H'_{L^c})}Ae^{-\lambda(H'_L+H'_{L^c})}\right\|,\\
&\|E^{H'_L+H'_{L^c}}[M,\infty) \,A\, E^{H'_L+H'_{L^c}}(-\infty,N]\|\\
&\quad\leq \exp(-\lambda(M-N-4\|H_{\partial L}\|)) \left\|e^{\lambda H_\Lambda}Ae^{-\lambda H_\Lambda}\right\|.
\end{align*}
\end{lemma}

\begin{remark}\label{rmk_hadamard_finite}
In the proofs of \cite[Theorems 2.1 and 5.1]{Arad-Kuwahara-Landau 2016}, the Hadamard formula is used to estimate imaginary-time conjugations of the form $e^{tH_\Lambda}Ae^{-tH_\Lambda}$. A direct Hadamard estimate with respect to the full Hamiltonian $H_\Lambda$ may yield an admissible decay parameter $\lambda$ depending on $(\partial L)^c$, rather than only on the boundary region.

The proof of Lemma \ref{lem_AKL2.1_5.1} avoids this difficulty by focusing on
\[
F_1(A,t)=e^{tH_\Lambda}e^{-t(H'_L+H'_{L^c})}A e^{t(H'_L+H'_{L^c})}e^{-tH_\Lambda}.
\]
This family satisfies a differential equation generated by the imaginary-time evolution of the boundary term $H_{\partial L}$. Thus $\lambda$ is independent of the total volume.

In the finite-volume setting, the same locality issue is also addressed in \cite[Supplementary Lemma 17]{Kuwahara-Saito2020}. The argument of Lemma \ref{lem_AKL2.1_5.1} may be viewed as an alternative proof of the corresponding finite-volume Hadamard estimate.
\end{remark}

A proof of this Lemma is given in Subsection \ref{S_AKL2.1_5.1}.

\begin{lemma}\label{lem_AKL6.2}
We set $\lambda$ as in Theorem \ref{thm_main_finite}. Then, for any $M,N\in\mathbb{R}$ and any $A\in\mathcal{A}_\Lambda$, we have
\begin{align*}
&\left\|E^{\bar{H}_\Lambda}[M,\infty) \,A\, E^{\bar{H}_\Lambda}(-\infty,N]\right\|\\
&\quad\leq \exp(-\lambda(M-N-8\|H_{\partial L}\|)) \left\|e^{\lambda\overline{H'_L+H'_{L^c}}}Ae^{-\lambda\overline{H'_L+H'_{L^c}}}\right\|,\\
&\left\|E^{\overline{H'_L+H'_{L^c}}}[M,\infty) \,A\, E^{\overline{H'_L+H'_{L^c}}}(-\infty,N]\right\|\\
&\quad\leq \exp\left(-\lambda(M-N-8\|H_{\partial L}\|)\right)\left\| e^{\lambda\bar{H}_\Lambda}Ae^{-\lambda\bar{H}_\Lambda}\right\|.
\end{align*}
\end{lemma}

Lemma \ref{lem_AKL6.2} can be proved by the same strategy as in \cite[Lemma 6.2]{Arad-Kuwahara-Landau 2016}. Our formulation yields a slightly sharper bound. Since no essentially new idea is involved, we defer the proof to Appendix \ref{S_AKL2}, where we prove a more general infinite-volume version.

\begin{lemma}\label{lem_AKL2.2_6.1}
For any $\epsilon\geq0$ and any $N>0$, we have
\begin{align*}
\left\|\Big(I-E^{H'_L+H'_{L^c}}(-N,N)\Big)\, E^{H_\Lambda}[- \epsilon,\epsilon]\right\|
&\leq 2\sqrt{2} \exp(-\lambda(N-2\epsilon-6\|H_{\partial L}\|)),\\
\left\|\Big(I-E^{\overline{H'_L+H'_{L^c}}}(-N,N)\Big)\, E^{\bar{H}_\Lambda}[-\epsilon,\epsilon]\right\|
&\leq 2\sqrt{2} \exp(-\lambda(N-2\epsilon-10\|H_{\partial L}\|)).
\end{align*}
In particular, by setting $N=M$, the second inequality yields
\[
\left\|\Big(I-E^{H'_L+H'_{L^c}}(-M,M)\Big)\, E^{\bar{H}_\Lambda}[-\epsilon,\epsilon]\right\| \leq 2\sqrt{2} \exp(-\lambda(M-2\epsilon-10\|H_{\partial L}\|)).
\]
\end{lemma}

Since $E^{\overline{H'_L+H'_{L^c}}}(-N,N)^c=0$ for $N>M$, the second inequality is meaningless in that case.

The proof in \cite[Theorem 2.2 and Theorem 6.1]{Arad-Kuwahara-Landau 2016} is essentially correct. However, we give a simpler proof based on a Markov inequality in
Subsection \ref{S_AKL2.2_6.1}.

\medskip
\begin{proof}[Proof of Theorem \ref{thm_main_finite}]
It is enough to show that
\begin{align*}
&\left\|E^{H_\Lambda}(-\delta+\xi,\delta+\xi)^c\, E^{\bar{H}_\Lambda}[-\epsilon+\xi,\epsilon+\xi]\right\|\\
&\quad\leq\frac{1}{\delta} \left(\epsilon+2\sqrt{2}(M+\|H_{\partial L}\|+|\xi|+\delta) \exp(-\lambda(M-2\epsilon-2|\xi|-10\|H_{\partial L}\|))\right).
\end{align*}
holds for any $\delta>\epsilon\geq0$ and any $\xi\in\mathbb{R}$. Indeed, part (ii) is the unshifted case $\xi=0$. To obtain part (i), we set
\[
(\epsilon,\delta,\xi)=(p+2\|H_{\partial L}\|,q+2\|H_{\partial L}\|,-2\|H_{\partial L}\|).
\]
Then $\delta>\epsilon\geq0$. Moreover,
\[
-\delta+\xi=-q-4\|H_{\partial L}\|<-2\|H_{\partial L}\|,\qquad
-\epsilon+\xi=-p-4\|H_{\partial L}\|\leq -2\|H_{\partial L}\|,
\]
while $\delta+\xi=q$ and $\epsilon+\xi=p$.  Since $H_\Lambda$ and $\bar{H}_\Lambda$ have no spectrum below $-2\|H_{\partial L}\|$, we have
\[
E^{H_\Lambda}(-\delta+\xi,\delta+\xi)
=E^{H_\Lambda}(-\infty,q),\quad E^{\bar{H}_\Lambda}[-\epsilon+\xi,\epsilon+\xi]
=E^{\bar{H}_\Lambda}(-\infty,p].
\]
Substituting these values into the shifted estimate gives part (i).

We start with the decomposition:
\begin{align*}
&\left\|\Big(I-E^{H_\Lambda}(-\delta+\xi,\delta+\xi)\Big)\, E^{\bar{H}_\Lambda}[-\epsilon+\xi,\epsilon+\xi]\right\|\\
&\leq\left\|\Big(I-E^{H_\Lambda}(-\delta+\xi,\delta+\xi)\Big)\, \Big(I-E^{H'_L+H'_{L^c}}(-M,M)\Big)\,
E^{\bar{H}_\Lambda}[-\epsilon+\xi,\epsilon+\xi]\right\|\\
&\quad+\left\|\Big(I-E^{H_\Lambda}(-\delta+\xi,\delta+\xi)\Big)\, E^{H'_L+H'_{L^c}}(-M,M)\, E^{\bar{H}_\Lambda}[-\epsilon+\xi,\epsilon+\xi]\right\|.
\end{align*}

To estimate the second term, we evaluate the following term:
\begin{align*}
&(H_\Lambda-\xi)\,E^{H'_L+H'_{L^c}}(-M,M)\,
E^{\bar{H}_\Lambda}[-\epsilon+\xi,\epsilon+\xi]\\
&\quad=(\bar{H}_\Lambda-\xi)\,E^{H'_L+H'_{L^c}}(-M,M)\,
E^{\bar{H}_\Lambda}[-\epsilon+\xi,\epsilon+\xi].
\end{align*}
Then we have
\begin{align*}
&\left\|(\bar{H}_\Lambda-\xi)\, E^{H'_L+H'_{L^c}}(-M,M)\, E^{\bar{H}_\Lambda}[-\epsilon+\xi,\epsilon+\xi]\right\|\\
&\leq\left\|(\bar{H}_\Lambda-\xi)\,E^{\bar{H}_\Lambda}[-\epsilon+\xi,\epsilon+\xi]\right\|\\
&\quad+\left\|(\bar{H}_\Lambda-\xi)\,
\Big(I-E^{H'_L+H'_{L^c}}(-M,M)\Big)\,E^{\bar{H}_\Lambda}[-\epsilon+\xi,\epsilon+\xi]\right\|\\
&\leq \epsilon+(M+\|H_{\partial L}\|+|\xi|) \left\|\left(I-E^{H'_L+H'_{L^c}}(-M,M)\right)
E^{\bar{H}_\Lambda}[-\epsilon+\xi,\epsilon+\xi]\right\|.
\end{align*}
By Markov inequality,
\begin{align*}
&\left\|\Big(I-E^{H_\Lambda}(-\delta+\xi,\delta+\xi)\Big)\, E^{H'_L+H'_{L^c}}(-M,M)\, E^{\bar{H}_\Lambda}[-\epsilon+\xi,\epsilon+\xi] \right\|\\
&\leq \frac{1}{\delta} \left(\epsilon+(M+\|H_{\partial L}\|+|\xi|) \left\|\left(I-E^{H'_L+H'_{L^c}}(-M,M)\right) E^{\bar{H}_\Lambda}[-\epsilon+\xi,\epsilon+\xi] \right\|\right)
\end{align*}

Consequently, we have
\[
\left\|\Big(I-E^{H_\Lambda}(-\delta+\xi,\delta+\xi)\Big)\, E^{\bar{H}_\Lambda}[-\epsilon+\xi,\epsilon+\xi]\right\|
\leq\frac{\epsilon+\eta(\epsilon,\delta,\xi)}{\delta},
\]
where
\begin{align*}
\eta(\epsilon,\delta,\xi):=
&(M+\|H_{\partial L}\|+|\xi|+\delta)\\
&\qquad\left\|\left(I-E^{H'_L+H'_{L^c}}(-M,M)\right) E^{\bar{H}_\Lambda}[-\epsilon+\xi,\epsilon+\xi] \right\|.
\end{align*}
Since
\begin{align*}
&\left\|\left(I-E^{H'_L+H'_{L^c}}(-M,M)\right) E^{\bar{H}_\Lambda}[-\epsilon+\xi,\epsilon+\xi]\right\|\\
&\leq \left\|\left(I-E^{H'_L+H'_{L^c}}(-M,M)\right) E^{\bar{H}_\Lambda}[-(\epsilon+|\xi|),\epsilon+|\xi|] \right\|,
\end{align*}
and applying Lemma \ref{lem_AKL2.2_6.1}, we obtain the desired inequality.
\end{proof}

\subsection{Proof of Lemma \ref{lem_AKL2.1_5.1}}\label{S_AKL2.1_5.1}
\begin{proof}
We will prove only the first inequality because the second inequality can be shown in the same manner.

For each $A\in\mathcal{A}_\Lambda$ and each $t\in\mathbb{R}$, we set
\[
F_1(A,t)=e^{tH_\Lambda}e^{-t\left(H'_L+H'_{L^c}\right)} Ae^{t\left(H'_L+H'_{L^c}\right)}e^{-tH_\Lambda}.
\]
Then we have
\begin{align*}
&\|E^{H_\Lambda}[M,\infty) \,A\, E^{H_\Lambda}(-\infty,N]\|\\
&\leq \|E^{H_\Lambda}[M,\infty) e^{-\lambda H_\Lambda}\|\,
\|F_1(e^{\lambda \left(H'_L+H'_{L^c}\right)}A e^{-\lambda \left(H'_L+H'_{L^c}\right)},\lambda)\|\,
\|e^{\lambda H_\Lambda}E^{H_\Lambda}(-\infty,N]\|.
\end{align*}

Here, $F_1(A,t)$ is the (unique)-solution of the differential equation:
\[
\frac{d}{dt} F_1(A,t)=[H_1(t),F_1(A,t)]
\quad\text{with}\quad F_1(A,0)=A,
\]
\[
H_1(t):=e^{tH_\Lambda}H_{\partial L}e^{-tH_\Lambda}.
\]
The unique solution is given by
\[
F_1(B,t)=W(t)\,B\,W(t)^{-1},
\]
where $W(t)$ is defined by the Dyson series (see, e.g., \cite[Proposition 2.2]{Nachtergaele-Sims-Young})
\[
W(t):=I+\sum_{n=1}^\infty
\int_0^t \int_0^{s_1} \cdots
\int_0^{s_{n-1}}
H_1(s_1) H_1(s_2) \cdots H_1(s_n)
\,ds_n \cdots ds_2\,ds_1.
\]
Moreover, we have
\begin{align*}
\max\{\|W(t)\|,\|W(t)^{-1}\|\}
\leq\exp\left(\int_0^t \|H_1(u)\|\,du\right)
\leq \exp\left(t\sup_{u\in[0,t]}\|H_1(u)\|\right).
\end{align*}

By Hadamard formula and since $u\in[0,\lambda]$, we have 
\[
\|H_1(u)\|
=\|e^{uH_\Lambda}H_{\partial L}e^{-uH_\Lambda}\|
\leq \frac{\|H_{\partial L}\|}{1-u/(2\lambda)}
\leq 2\|H_{\partial L}\|.
\]
where we note that $\lambda$ depends only on $\partial L$, since $H_{\partial L}$ is located on $\partial L$.

Consequently, we obtain
\[
\|F_1(e^{\lambda(H'_L+H'_{L^c})}A e^{-\lambda(H'_L+H'_{L^c})},\lambda)\|
\leq \exp(4\lambda\|H_{\partial L}\|) \left\|e^{\lambda(H'_L+H'_{L^c})}A e^{-\lambda(H'_L+H'_{L^c})}\right\|.
\]
Hence, we obtain the desired formula.
\end{proof}

\subsection{Proof of Lemma \ref{lem_AKL2.2_6.1}}\label{S_AKL2.2_6.1}
\begin{proof}
We will show the first inequality. For any $\psi\in\mathcal{H}_\Lambda$, we define
\[
\phi:=E^{H_\Lambda}[-\epsilon,\epsilon]\, \Big(I-E^{H'_L+H'_{L^c}}(-N,N)\Big)\psi.
\]
Then we have
\[
\|(H'_L+H'_{L^c})\phi\|\leq(\epsilon+\|H_{\partial L}\|)\|\phi\|.
\]
Hence by Markov inequality, we obtain
\[
\left\|\Big(I-E^{H'_L+H'_{L^c}}(-\delta,\delta)\Big)\phi\right\|
\leq\frac{\epsilon+\|H_{\partial L}\|}{\delta}\|\phi\|.
\]
In particular, choosing $\delta=2(\epsilon+\|H_{\partial L}\|)$, we have
\[
\frac{1}{2}\|\phi\|^2
\leq\|\phi\|^2-\left\|\Big(I-E^{H'_L+H'_{L^c}}(-\delta,\delta)\Big)\phi\right\|^2
=\left\|E^{H'_L+H'_{L^c}}(-\delta,\delta)\phi\right\|^2
\]
and thus
\[
\|\phi\|\leq\sqrt{2}\left\|E^{H'_L+H'_{L^c}}(-\delta,\delta)\phi\right\|.
\]

Here, since
\[
e^{\lambda H_\Lambda}E^{H_\Lambda}[-\epsilon,\epsilon]e^{-\lambda H_\Lambda}=E^{H_\Lambda}[-\epsilon,\epsilon]
\]
and by Lemma \ref{lem_AKL2.1_5.1}, we have
\begin{align*}
\|\phi\|
&\leq \sqrt{2}\left\|
E^{H'_L+H'_{L^c}}(-\infty,\delta)\,
E^{H_\Lambda}[-\epsilon,\epsilon]\,
E^{H'_L+H'_{L^c}}[N,\infty) \psi\right\|\\
&\qquad+\sqrt{2}\left\|
E^{H'_L+H'_{L^c}}(-\delta,\infty)\,
E^{H_\Lambda}[-\epsilon,\epsilon]\,
E^{H'_L+H'_{L^c}}(-\infty,-N]\psi\right\|\\
&\leq 2\sqrt{2}\exp(-\lambda(N-\delta-4\|H_{\partial L}\|))\|\psi\|\\
&=2\sqrt{2}\exp(-\lambda(N-2\epsilon-6\|H_{\partial L}\|))\|\psi\|.
\end{align*}

We can show the second inequality in the same way using Lemma \ref{lem_AKL6.2}.
\end{proof}

\newpage
\section{Infinite Systems}\label{S_infinite}
\subsection{Setup}
Let
\(
\mathcal{A}_{\mathrm{loc}}
:=\bigcup_{\Lambda\Subset\Gamma}\mathcal{A}_\Lambda
\)
be the local algebra, and let $\mathcal{A}$ denote its norm closure. Under Assumption \ref{ass_interaction}, the thermodynamic limit
\begin{equation}\label{eq_dynamic}
\alpha_t^\Gamma(A)
:=\lim_{\substack{\Lambda\Subset\Gamma\\ \Lambda\to\infty}}\alpha_t^\Lambda(A)
\quad(A\in\mathcal{A}_{\mathrm{loc}},\ t\in\mathbb{R}),
\end{equation}
exists in norm \cite[Theorem 3.5]{Nachtergaele-Sims-Young}.

We define the linear maps \(\delta_\Gamma,\delta_{X^c}:\mathcal{A}_{\mathrm{loc}}\to
\mathcal{A}_{\mathrm{loc}}\) by
\[
\delta_\Gamma(A):=\sum_{Z\Subset\Gamma}[\Phi(Z),A],\quad
\delta_{X^c}(A):=\sum_{\substack{Z\Subset\Gamma\\ Z\cap X^c\neq\emptyset}}[\Phi(Z),A].
\]

\begin{assumption}
We assume that $\omega$ is a ground state for this dynamics, that is,
\[
\omega\left(A^*\delta_\Gamma(A)\right)\geq 0
\quad (A\in\mathcal{A}_{\mathrm{loc}}).
\]
\end{assumption}
Let $(\mathcal{H}_\omega,\pi_\omega,\Omega)$ be the GNS triple associated with $\omega$. Then, there exists a (possibly unbounded) self-adjoint operator $H_\omega$ on $\mathcal{H}_\omega$ such that
\[
\pi_\omega(\alpha_t^\Gamma(A))=e^{itH_\omega}\pi_\omega(A)e^{-itH_\omega},\quad H_\omega\geq0,\quad H_\omega\Omega=0.
\]

Fix a finite set $X\Subset\Gamma$, corresponding to the boundary region $\partial L$ in \eqref{eq_decompose}. We define
\[
H_X:=\sum_{Y\subset X}\Phi(Y)\in\mathcal{A}_X,\ 
H_{X^c}:=H_\omega-\pi_\omega(H_X),\ 
\operatorname{Dom}(H_{X^c}):=\operatorname{Dom}(H_\omega),
\]
where we denote by $\operatorname{Dom}(H_\omega)$ the domain of the unbounded operator $H_\omega$. 

For any self-adjoint operator $H$ and any Borel set $I\subset\mathbb R$, we write $E^H(I)$ for the corresponding spectral projection. Equivalently,
\[
H=\int_{-\infty}^{\infty}\lambda\,E^H(d\lambda).
\]
We use the standard spectral functional calculus for self-adjoint operators; see, e.g., \cite[Section 10.1, Proposition 10.1]{Hall:book}. In particular, for $z\in\mathbb{C}$, the operators $e^{zH_\omega}$ and $e^{zH_{X^c}}$ are understood in this sense, and
\[
\mathcal{D}_H:=\bigcup_{N\in\mathbb{N}}E^H[-N,N]\mathcal{H}_\omega\subset \operatorname{Dom}(e^{zH})\quad (H=H_\omega,H_{X^c}).
\]

\begin{definition}[energy-truncated Hamiltonian]
For any $M>\|H_X\|$, we define
\begin{align*}
&\bar{H}_{X^c}:= H_{X^c}E^{H_{X^c}}(-M,M)+M\,E^{H_{X^c}}[M,\infty),\\
&\bar{H}_\omega:=\pi_\omega(H_X)+\bar H_{X^c}.
\end{align*}
\end{definition}
Here, we note that
\[
H_{X^c}=H_\omega-\pi_\omega(H_X)\geq -\|H_X\|,
\]
\[
\bar{H}_\omega\leq H_\omega \quad\text{and}\quad
\|\bar{H}_\omega\|\leq\|H_X\|+M.
\]

\subsection{Main Theorem}
\begin{theorem}\label{thm_main_infinite}
We set
\[
\lambda:=\frac{1}{4\mathfrak{j}_\Phi(N_\Phi+|X|)}, \quad N_\Phi:=\sup_{x\in\Gamma} |\{ y\in\Gamma \,;\, d(x,y)\leq \mathfrak{r}_\Phi\}|<\infty,
\]
\begin{align*}
\delta(p,q)&:=2\sqrt{2}(M+5\|H_X\|+q) \exp(-\lambda(M-2p-18\|H_X\|)),\\
\eta(\epsilon,\delta)&:=2\sqrt{2}(M+\|H_X\|+\delta) \exp(-\lambda(M-2\epsilon-10\|H_X\|)).
\end{align*}
\begin{itemize}
\item[(i)] For any $q>p\geq-2\|H_X\|$, we have
\[
\left\|\left(I-E^{H_\omega}(-\infty,q)\right)\,
E^{\bar{H}_\omega}(-\infty,p]\right\|
\leq\frac{p+2\|H_X\|+\delta(p,q)}{q+2\|H_X\|}.
\]
\item[(ii)] For any $\delta>\epsilon\geq0$, we have
\[
\left\|\left(I-E^{H_\omega}(-\delta,\delta)\right)\,
E^{\bar{H}_\omega}[-\epsilon,\epsilon]\right\|
\leq\frac{\epsilon+\eta(\epsilon,\delta)}{\delta}.
\]
\end{itemize}
\end{theorem}

In the infinite-volume statement, the constants depend on the local interaction parameters and on the fixed finite region $X$, but not on any approximating finite volume.

\begin{corollary}\label{cor_main2_infinite}
For any $j=0,1,2,\ldots$, define the following rank thresholds of the spectral projections:
\begin{align*}
\epsilon_j
&:=\inf\left\{p\in\mathbb{R}\,\middle|\,
\operatorname{Rank}\left(E^{H_\omega}(-\infty,p]\right)\geq j+1\right\},\\
\bar{\epsilon}_j
&:=\inf\left\{p\in\mathbb{R}\,\middle|\,
\operatorname{Rank}\left(E^{\bar{H}_\omega}(-\infty,p]\right)\geq j+1\right\}.
\end{align*}
When the relevant low-energy spectrum is discrete, isolated, and of finite multiplicity, these quantities coincide with the usual low-lying eigenvalues counted with multiplicity.

We also set
\begin{align*}
\delta_j
&:=2\sqrt{2}(M+5\|H_X\|+\epsilon_j)
\exp(-\lambda(M-2\epsilon_j-18\|H_X\|)),\\
\eta
&:=\eta(2\delta_0,\Delta)
=2\sqrt{2}(M+\|H_X\|+\Delta)
\exp(-\lambda(M-4\delta_0-10\|H_X\|)).
\end{align*}
Assume that
\[
\delta_j>\delta(\epsilon_j-2\delta_j,\epsilon_j).
\]

\begin{itemize}
\item[(i)] We have
\(
\epsilon_j-2\delta_j\leq \bar{\epsilon}_j\leq \epsilon_j.
\)
\item[(ii)] If $\omega$ is a locally unique ground state with spectral gap $\Delta>2\delta_0+\eta>0$: 
\[
\omega(A^*\delta_\Gamma(A))\geq\Delta\,\omega(A^*A)
\quad(A\in\mathcal{A}_{\mathrm{loc}} \text{ with } \omega(A)=0),
\]
then there exists a unit vector $\bar{\Omega}\in
E^{\bar{H}_\omega}[-2\delta_0,2\delta_0]\mathcal H_\omega$
such that
\[
|\langle\Omega,\bar{\Omega}\rangle|^2\geq 1-\left(\frac{2\delta_0+\eta}{\Delta}\right)^2.
\]
\end{itemize}
\end{corollary}

Corollary \ref{cor_main2_infinite} follows from Theorem \ref{thm_main_infinite} by the same argument as in the proof of Theorem \ref{thm_main2_finite}.

\subsection{Lemmas and Proof of Theorem \ref{thm_main_infinite}}
Additional care is required to extend Lemma \ref{lem_AKL2.1_5.1} to the infinite-volume setting. Once this extension is established, the remaining arguments follow essentially in the same way as in finite volume.

\begin{proposition}[Infinite extension of Lemma \ref{lem_AKL2.1_5.1}]\label{prop_AKL}
We set $\lambda$ as in Theorem \ref{thm_main_infinite}. We assume that $A,B\in\mathcal{B}(\mathcal{H}_\omega)$ are such that, for every $s\in\mathbb{C}$, the sesquilinear forms
\[
(\phi,\psi)\in\mathcal{D}_{H_{X^c}}\times\mathcal{D}_{H_{X^c}}
\mapsto
\left\langle e^{\bar{s} H_{X^c}}\phi,
A e^{-sH_{X^c}}\psi\right\rangle,
\]
\[
(\phi,\psi)\in\mathcal D_{H_\omega}\times \mathcal D_{H_\omega}
\mapsto
\left\langle e^{\bar{s} H_\omega}\phi,
B e^{-sH_\omega} \psi\right\rangle
\]
extend to bounded sesquilinear forms on $\mathcal{H}_\omega\times\mathcal{H}_\omega$. We also assume that the corresponding bounded operators, denoted by
\[
e^{sH_{X^c}}Ae^{-sH_{X^c}},\quad e^{sH_\omega}Be^{-sH_\omega},
\]
define $\mathcal{B}(\mathcal H_\omega)$-valued entire functions of $s$.

Then for any $M,N\in\mathbb{R}$, we have
\begin{align*}
\|E^{H_\omega}[M,\infty) \,A\,
E^{H_\omega}(-\infty,N]\|
&\leq\exp(-\lambda(M-N-4\|H_X\|))\,
\|e^{\lambda H_{X^c}}Ae^{-\lambda H_{X^c}}\|,\\
\|E^{H_{X^c}}[M,\infty) \,B\,
E^{H_{X^c}}(-\infty,N]\|
&\leq\exp(-\lambda(M-N-4\|H_X\|))
\|e^{\lambda H_\omega}Be^{-\lambda H_\omega}\|.
\end{align*}
\end{proposition}

We will give a proof in Subsection \ref{SS_prop_AKL}. Additional care is needed in infinite volume, since $H_\omega$ and $H_{X^c}$ are generally unbounded and the formal expression
\[
F_1(A,t)=e^{tH_\omega}e^{-tH_{X^c}} Ae^{tH_{X^c}}e^{-tH_\omega} \quad(A\in\mathcal{B}(\mathcal{H}_\omega),\ t\in\mathbb{R})
\]
is not a bounded-operator identity without an analytic or domain
justification.

In the application below, Proposition \ref{prop_AKL} is used with $A=E^{H_{X^c}}[-\epsilon,\epsilon]$ and $B=E^{H_\omega}[-\epsilon,\epsilon]$. These projections commute with $H_{X^c}$ and $H_\omega$, respectively. Hence the required operator-valued functions are constant, and in particular entire.

\begin{remark}\label{rmk_hadamard_infinite}
A related infinite-volume estimate is stated in \cite[Lemma C.6]{Liu-Yi-Zhou-Zou2025}. Its proof uses, in \cite[Eq. (C.89)-(C.90)]{Liu-Yi-Zhou-Zou2025},
\[
\operatorname{Ad}_{e^{\nu H'_t}}(O_s)
=e^{\nu\operatorname{ad}_{H'_t}}(O_s)
\]
together with a norm estimate on the Hadamard series. This step requires that $O_s$ be an analytic element for the derivation generated by $H'_t$, or an equivalent domain/commutator justification. Such a property is not automatic when $O_s$ is a spectral projection associated with a non-compact block.
\end{remark}

\begin{lemma}[Infinite extension of Lemma \ref{lem_AKL6.2}]\label{lem_AKL2}
We set $\lambda$ as in Theorem \ref{thm_main_infinite}. Then for any $A,B\in\mathcal{B}(\mathcal{H}_\omega)$ and any $M,N\in\mathbb{R}$, we have
\[
\left\|E^{\bar{H}_\omega}[M,\infty) \,A\, E^{\bar{H}_\omega}(-\infty,N]\right\|
\leq \exp(-\lambda(M-N-8\|H_X\|)) \left\|e^{\lambda\bar{H}_{X^c}}Ae^{-\lambda\bar{H}_{X^c}}\right\|,
\]
\[
\left\|E^{\bar{H}_{X^c}}[M,\infty) \,B\, E^{\bar{H}_{X^c}}(-\infty,N]\right\|
\leq \exp\left(-\lambda(M-N-8\|H_X\|)\right)\left\| e^{\lambda\bar{H}_\omega}Be^{-\lambda\bar{H}_\omega}\right\|.
\]
\end{lemma}

We can show this lemma with the same strategy as \cite[Lemma 6.2]{Arad-Kuwahara-Landau 2016}. We give the proof in Appendix \ref{S_AKL2}.

\begin{lemma}\label{lem_AKL3}
For any $\epsilon\geq0$ and any $N>0$, we have
\begin{align*}
\left\|E^{H_{X^c}}(-N,N)^c\,
E^{H_\omega}[-\epsilon,\epsilon]\right\|
&\leq 2\sqrt{2} \exp(-\lambda(N-2\epsilon-6\|H_X\|)),\\
\left\|E^{\bar{H}_{X^c}}(-N,N)^c\,
E^{\bar{H}_\omega}[-\epsilon,\epsilon]\right\|
&\leq 2\sqrt{2} \exp(-\lambda(N-2\epsilon-10\|H_X\|)).
\end{align*}
In particular, by setting $N=M$, the second inequality yields
\[
\left\|E^{H_{X^c}}(-M,M)^c\,
E^{\bar{H}_\omega}[-\epsilon,\epsilon]\right\|
\leq 2\sqrt{2} \exp(-\lambda(M-2\epsilon-10\|H_X\|)).
\]
\end{lemma}

\begin{proof}
We can show in the same way as Lemma \ref{lem_AKL2.2_6.1}, replacing
\[
H_\Lambda=(H'_L+H'_{L^c})+H_{\partial L},\quad
\bar{H}_\Lambda=\overline{H'_L+H'_{L^c}}+H_{\partial L}
\]
to
\[
H_\omega=H_{X^c}+\pi_\omega(H_X),\quad
\bar{H}_\omega=\bar{H}_{X^c}+\pi_\omega(H_X), 
\]
respectively, and using Proposition \ref{prop_AKL} and Lemma \ref{lem_AKL2}.
\end{proof}

\begin{proof}[Proof of Theorem \ref{thm_main_infinite}]
We can show in the same way as Theorem \ref{thm_main_finite} using Lemma \ref{lem_AKL3}.
\end{proof}

\subsection{Proof of Proposition \ref{prop_AKL}}\label{SS_prop_AKL}
\begin{lemma}\label{lem_derivation}
We set $\lambda$ as in Theorem \ref{thm_main_infinite} and define
\[
S_X:=2\lambda,\quad R(S_X):=\{t\in\mathbb{C}\,;\,|t|<S_X\}.
\]
\begin{itemize}
\item[(i)] The series
\begin{align*}
\exp(s\delta_\Gamma)(H_X)
&:=\sum_{n=0}^\infty\frac{s^n}{n!} (\delta_\Gamma)^n(H_X),\\
\exp(s\delta_{X^c})(H_X)
&:=\sum_{n=0}^\infty\frac{s^n}{n!} (\delta_{X^c})^n(H_X)
\end{align*}
are absolutely convergent for any $s\in R(S_X)$. Moreover, we then have
\[
\max\{\|\exp(s\delta_\Gamma)(H_X)\|, \|\exp(s\delta_{X^c})(H_X)\|\}
\leq \frac{\|H_X\|}{1-|s|/S_X}.
\]

\item[(ii)] For any $s\in R(S_X)$ and any $A\in\mathcal{B}(\mathcal{H}_\omega)$, the following
$\mathcal{B}(\mathcal{H}_\omega)$-valued initial value problems posed on the interval $t\in[0,1]$ admit unique solutions:
\[
\frac{d}{dt} F_1(s,A,t)=[s\exp(ts\delta_\Gamma)(H_X),F_1(s,A,t)]
\quad\text{with}\quad F_1(s,A,0)=A,
\]
\[
\frac{d}{dt} F_2(s,A,t)=[s\exp(ts\delta_{X^c})(-H_X),F_2(s,A,t)]
\quad\text{with}\quad F_2(s,A,0)=A.
\]
Moreover, we then have
\[
\max\{\|F_1(s,A,t)\|,\|F_2(s,A,t)\|\}
\leq \exp\left(2|s|\frac{\|H_X\|}{1-t|s|/S_X}\right) \|A\|.
\]
\end{itemize}
\end{lemma}

\begin{proof}
Part (i) follows from the bound (see, e.g. \cite[Lemma 3.3.11]{Naaijkens:book})
\begin{align*}
&\max\left\{
\left\|\frac{s^n}{n!}(\delta_\Gamma)^n(H_X)\right\|,
\left\|\frac{s^n}{n!}(\delta_{X^c})^n(H_X)\right\|
\right\}\\
&\leq \frac{(2\mathfrak{j}_\Phi|s|)^n}{n!} \|H_X\| \prod_{k=1}^n(N_\Phi k+|X|)
\leq (2\mathfrak{j}_\Phi|s|(N_\Phi+|X|))^n \|H_X\|.
\end{align*}

\medskip
Part (ii). The unique solution \(F_1\) is given by
\[
F_1(s,A,t)=W(t)\,A\,W(t)^{-1},
\]
where \(W(t)\) is defined by the Dyson series (see, e.g., \cite[Proposition 2.2]{Nachtergaele-Sims-Young})
\[
W(t):=I+\sum_{n=1}^\infty
\int_0^t \int_0^{s_1} \cdots
\int_0^{s_{n-1}}
H_1(s_1) H_1(s_2) \cdots H_1(s_n)
\,ds_n \cdots ds_2\,ds_1,
\]
and
\[
H_1(t):=s\exp(ts\delta_\Gamma)(H_X),\quad t\in[0,1].
\]
Moreover, the claimed bound follows from part (i):
\begin{align*}
&\max\{\|W(t)\|,\|W(t)^{-1}\|\}
\leq \exp\left(\sup_{u\in[0,t]} \|H_1(u)\|\right)
\leq \exp\left(\frac{|s|\,\|H_X\|}{1-t|s|/S_X}\right).
\end{align*}
Hence the desired bound is obtained.

The same argument yields the bound for $F_2$.
\end{proof}

\medskip
\begin{lemma}\label{lem_infinite}
Assume that $A,B\in\mathcal{B}(\mathcal{H}_\omega)$ satisfy the analytic-extension hypothesis stated in Proposition \ref{prop_AKL}. In the identities below, the conjugations
\[
e^{-tsH_{X^c}}Ae^{tsH_{X^c}},\quad e^{-tsH_\omega}Be^{tsH_\omega}
\]
are understood in the sense of this analytic-extension hypothesis. Then, for any $(t,s)\in [0,1]\times R(S_X)$,
\[
\langle\phi,F_1(s,A,t)\psi\rangle=\langle e^{t\bar{s}H_\omega}\phi,\, e^{-ts H_{X^c}}A e^{ts H_{X^c}}e^{-tsH_\omega}\psi\rangle
\]
holds for any
\[
\psi,\phi\in\mathcal{D}_{H_\omega}:= \bigcup_{N\in\mathbb{N}} E^{H_\omega}[-N,N]\mathcal{H}_\omega,
\]
and
\[
\langle\phi,F_2(s,B,t)\psi\rangle=\langle e^{t\bar{s}H_{X^c}}\phi,\,e^{-tsH_\omega}Be^{tsH_\omega}e^{-tsH_{X^c}}\psi\rangle
\]
holds for any
\[
\psi,\phi\in\mathcal{D}_{H_{X^c}}
:=\bigcup_{N\in\mathbb{N}} E^{H_{X^c}}[-N,N]\mathcal{H}_\omega.
\]
\end{lemma}
\begin{proof}
For any $s\in\mathbb{C}$, define
\[
f(s):=\left\langle e^{t\bar{s}H_\omega}\phi,\,
\left(e^{-tsH_{X^c}}Ae^{tsH_{X^c}}\right)
e^{-tsH_\omega}\psi\right\rangle,
\]
where the middle conjugation is understood in the sense of the analytic-extension hypothesis. Then $f$ is analytic on $\mathbb{C}$.

For any $s\in R(S_X)$, we also define
\[
g(s):=\langle\phi,F_1(s,A,t)\psi\rangle,
\]
which is analytic on $R(S_X)$ by Lemma \ref{lem_derivation}. To conclude that $f(s)=g(s)$ for all $s\in R(S_X)$, it is enough to show that $f(is)=g(is)$ for any $s\in(-S_X,S_X)$ by the identity theorem.

By equation \eqref{eq_dynamic}, for any $t\in\mathbb{R}$
\[
\alpha_t^\Gamma(H_X)
:=\lim_{\substack{\Lambda\Subset\Gamma\\ \Lambda\to\infty}} \alpha_t^\Lambda(H_X)
=\lim_{\substack{\Lambda\Subset\Gamma\\ \Lambda\to\infty}} \exp(it\delta_\Lambda)(H_X).
\]
By Lemma \ref{lem_derivation}(i), for any $(t,s)\in[0,1]\times(-S_X,S_X)$,
\[
\lim_{\substack{\Lambda\Subset\Gamma\\ \Lambda\to\infty}} \exp(its\delta_\Lambda)(H_X)
=\exp(its\delta_\Gamma)(H_X).
\]
Hence, for any $(t,s)\in[0,1]\times(-S_X,S_X)$, we have
\begin{align*}
\frac{d}{dt}e^{itsH_{X^c}} e^{-itsH_\omega}\psi
&=e^{itsH_{X^c}}e^{-itsH_\omega}\,
e^{itsH_\omega}(-isH_X)e^{-itsH_\omega}\psi\\
&=e^{itsH_{X^c}}e^{-itsH_\omega}\,
\exp(its\delta_\Gamma)(-isH_X)\psi.
\end{align*}
Since $\psi,\phi\in\mathcal{D}_{H_\omega}$ are arbitrary and $\mathcal{D}_{H_\omega} \subset\mathcal{H}_\omega$ is dense, this implies, for any $(t,s)\in[0,1]\times(-S_X,S_X)$,
\[
e^{itsH_{X^c}}e^{-itsH_\omega}
=I+\int_0^t e^{iusH_{X^c}}e^{-iusH_\omega} \exp(ius\delta_\Gamma)(-isH_X) \,du.
\]
In the same manner, we also obtain
\[
e^{itsH_\omega}e^{-itsH_{X^c}}
=I+\int_0^t \exp(ius\delta_\Gamma)(isH_X) e^{ius H_\omega}e^{-iusH_{X^c}}\,du.
\]

Consequently, if we set
\[
\hat{F}_1(t,s):=e^{itsH_\omega}e^{-itsH_{X^c}} A e^{itsH_{X^c}}e^{-itsH_\omega} \in\mathcal{B}(\mathcal{H}_\omega),
\]
then, for any $(t,s)\in[0,1]\times(-S_X,S_X)$,
\[
\frac{d}{dt} \hat{F}_1(t,s)
=[is\exp(its\delta_\Gamma)(H_X),\hat{F}_1(t,s)]
\in\mathcal{B}(\mathcal{H}_\omega).
\]
Since $\hat{F}_1(0,s)=A$, by Lemma \ref{lem_derivation}(ii), we obtain
\[
\hat{F}_1(t,s)=F_1(is,A,t).
\]
Hence, for any $(t,s)\in[0,1]\times(-S_X,S_X)$,
\[
f(is)=\langle\phi,\hat{F}_1(t,s)\psi\rangle=\langle\phi,F_1(is,A,t)\psi\rangle=g(is).
\]

The identity for $F_2$ is proved in the same way.
\end{proof}

\begin{proof}[Proof of Proposition \ref{prop_AKL}]
For any $\psi,\phi\in\mathcal{D}_{H_\omega}$, we have
\[
e^{-\lambda H_\omega}E^{H_\omega}[M,\infty)\phi,\quad e^{\lambda H_\omega}E^{H_\omega}(-\infty,N] \psi\in\mathcal{D}_{H_\omega}.
\]
By translating the entire function in the analytic-extension hypothesis, $e^{\lambda H_{X^c}}Ae^{-\lambda H_{X^c}}$ satisfies the same hypothesis. Thus by Lemma \ref{lem_infinite},
\begin{align*}
&\langle\phi,E^{H_\omega}[M,\infty) \,A\,
E^{H_\omega}(-\infty,N]\psi\rangle\\
&=\langle e^{-\lambda H_\omega}E^{H_\omega}[M,\infty)\phi,\,
F_1(\lambda,e^{\lambda H_{X^c}}A e^{-\lambda H_{X^c}},1)\,
e^{\lambda H_\omega}E^{H_\omega}(-\infty,N] \psi\rangle.
\end{align*}
By Lemma \ref{lem_derivation}, it follows that
\begin{align*}
&|\langle e^{-\lambda H_\omega}E^{H_\omega}[M,\infty)\phi,\,
F_1(\lambda,e^{\lambda H_{X^c}}A e^{-\lambda H_{X^c}},1)\,
e^{\lambda H_\omega}
E^{H_\omega}(-\infty,N] \psi\rangle|\\
&\leq \|e^{-\lambda H_\omega}E^{H_\omega}[M,\infty)\|\,
\|F_1(\lambda,e^{\lambda H_{X^c}}A e^{-\lambda H_{X^c}},1)\|\,
\|e^{\lambda H_\omega} E^{H_\omega}(-\infty,N]\|
\,\|\phi\|\,\|\psi\|\\
&\leq \exp(-\lambda(M-N-4\|H_X\|)) \|e^{\lambda H_{X^c}}A e^{-\lambda H_{X^c}}\| \,\|\phi\|\,\|\psi\|.
\end{align*}
Since $\psi,\phi\in\mathcal{D}_{H_\omega}$ is arbitrary and $\mathcal{D}_{H_\omega}\subset\mathcal{H}_\omega$ is dense, the first estimate follows.

The second estimate can be shown in the same way.
\end{proof}

\newpage
\appendix
\section{A Counterexample to a Volume-Uniform Operator-Norm Formulation}\label{S_counterexample}

In this section, we give a simple example showing that the operator-norm estimate underlying the volume-uniform interpretation of \cite{Arad-Kuwahara-Landau 2016} cannot, in general, be uniform in the total volume. Thus the finite-volume estimate may be valid for each fixed system size, but the constants cannot in general be chosen independently of the total volume.

We consider the nearest-neighbor ferromagnetic Ising model. Let $N$ be a positive odd integer and set
\[
\Gamma:=\{1,2,\ldots,2N\}\subset\mathbb{Z}, \qquad\mathcal{A}:= \bigotimes_{x\in\Gamma}M_2(\mathbb{C}).
\]
Let
\[
\sigma^Z:=
\begin{pmatrix}
1&0\\
0&-1
\end{pmatrix}
\in M_2(\mathbb{C}),
\]
and let $\sigma_j^Z\in\mathcal{A}_\Gamma$ denote the Pauli $Z$-matrix acting on the site $j\in\Gamma$. We define
\[
h_j:=I-\sigma_j^Z\sigma_{j+1}^Z,\qquad j=1,2,\ldots,2N-1,
\]
and set
\[
\Phi(Z):=
\begin{cases}
h_j & \text{if } Z=\{j,j+1\}
\text{ for some }j=1,\ldots,2N-1,\\
0 & \text{otherwise}.
\end{cases}
\]
Then
\[
H_\Gamma:=\sum_{Z\subset\Gamma}\Phi(Z)
=\sum_{j=1}^{2N-1}h_j.
\]
This is a bounded finite-range interaction with
\[
\mathfrak{r}_\Phi=1,\qquad
\mathfrak{j}_\Phi\leq 4,\qquad
N_\Phi:=\sup_{x\in\Gamma}
\left|\{y\in\Gamma\,;\, d(x,y)\leq \mathfrak{r}_\Phi\}\right|\leq3.
\]

We cut the chain between $N$ and $N+1$. Namely, we set
\[
L:=\{1,\ldots,N\},\qquad
L^c:=\{N+1,\ldots,2N\},\qquad
\partial L:=\{N,N+1\}.
\]
With the decomposition used in the present paper, we have
\[
H_\Gamma=H'_L+H_{\partial L}+H'_{L^c},
\]
where
\[
H'_L=\sum_{j=1}^{N-1}h_j,\qquad
H_{\partial L}=h_N,\qquad
H'_{L^c}=\sum_{j=N+1}^{2N-1}h_j.
\]
We fix $M>\|H_{\partial L}\|=2$, independently of $N$. The truncated Hamiltonian is
\[
\overline{H'_L+H'_{L^c}}
:=(H'_L+H'_{L^c})E^{H'_L+H'_{L^c}}(-M,M)
+M E^{H'_L+H'_{L^c}}[M,\infty),
\]
and
\[
\bar{H}_\Gamma:=
H_{\partial L}+\overline{H'_L+H'_{L^c}}.
\]
Since $H'_L+H'_{L^c}\geq0$, the projection $E^{H'_L+H'_{L^c}}(-M,M)$ is the same as $E^{H'_L+H'_{L^c}}[0,M)$.

\medskip
Equation (16) in \cite{Arad-Kuwahara-Landau 2016} is stated for the truncated Hamiltonian obtained by truncating $H_L$. The corresponding operator-norm estimate for the truncation used in the present paper would take the form
\begin{equation}\label{eq_AKL_original}
\left\|(H_\Gamma-\bar{H}_\Gamma) E^{\bar{H}_\Gamma}[0,\epsilon]\right\|
\leq \frac{6}{\lambda_{\mathrm{AKL}}^{3/2}} \exp\left\{-\lambda_{\mathrm{AKL}}\bigl(M-\epsilon_0(H'_L+H'_{L^c})-(\epsilon-\bar\epsilon_0)-33|\partial L|_{\mathrm{AKL}} \bigr)\right\},
\end{equation}
where
\[
\lambda_{\mathrm{AKL}}:=\frac{1}{2\mathfrak{j}_\Phi N_\Phi},\qquad |\partial L|_{\mathrm{AKL}}=\|H_{\partial L}\|,
\]
and where $\epsilon_0(H'_L+H'_{L^c})$ and $\bar{\epsilon}_0$ denote the lowest eigenvalues of $H'_L+H'_{L^c}$ and $\bar{H}_\Gamma$, respectively.

In the present Ising example, all interactions $\Phi(X)$ are non-negative. Moreover, the fully aligned state
\[
\Omega_+:=|\uparrow,\uparrow,\ldots,\uparrow\,\rangle,\qquad
|\uparrow\,\rangle:=
\begin{pmatrix}
1\\
0
\end{pmatrix},
\qquad
|\downarrow\,\rangle:=
\begin{pmatrix}
0\\
1
\end{pmatrix},
\]
satisfies $h_j\Omega_+=0$ for any $j$. Hence
\[
\epsilon_0(H'_L+H'_{L^c})=0.
\]
The truncation preserves this zero eigenvalue, and since $H_{\partial L}=h_N\geq0$, we also have
\[
\bar{\epsilon}_0=0.
\]
Therefore, taking $\epsilon=M$ in \eqref{eq_AKL_original} gives
\begin{equation}\label{eq_AKL_original_2}
\left\|(H_\Gamma-\bar{H}_\Gamma) E^{\bar{H}_\Gamma}[0,M]\right\|
\leq\frac{6}{\lambda_{\mathrm{AKL}}^{3/2}} e^{66\lambda_{\mathrm{AKL}}},
\end{equation}
whose right-hand side is independent of the total volume $|\Gamma|=2N$.

\medskip
We now show that such a bound cannot hold uniformly in $N$. Since $N$ is odd, the alternating spin configuration
\[
\chi:=|\uparrow\,\downarrow\,\uparrow\,\cdots\,\downarrow\,\uparrow\,\rangle\in \bigotimes_{j=1}^{N}\mathbb{C}^2.
\]
starts and ends with $|\uparrow\,\rangle$. We set
\[
\psi:=\chi\otimes\chi\in \left(\bigotimes_{j=1}^{N}\mathbb{C}^2\right) \otimes \left(\bigotimes_{j=N+1}^{2N}\mathbb{C}^2\right).
\]
Then the two spins across the cut, at sites $N$ and $N+1$, are both $|\uparrow\,\rangle$, while all other neighboring spins are anti-aligned. Consequently,
\[
h_N\psi=0,\qquad
h_j\psi=2\psi\quad (j\neq N).
\]
It follows that
\[
(H'_L+H'_{L^c})\psi=(4N-4)\psi,\qquad
H_{\partial L}\psi=0,
\]
and hence
\[
H_\Gamma\psi=(4N-4)\psi.
\]

For sufficiently large $N$ such that $4N-4>M$, we have
\[
\overline{H'_L+H'_{L^c}}\psi=M\psi,\qquad
\bar{H}_\Gamma\psi=M\psi.
\]
Thus
\[
\psi\in\operatorname{Ran}E^{\bar{H}_\Gamma}(-\infty,M]=\operatorname{Ran}E^{\bar{H}_\Gamma}[0,M],
\]
and
\[
(H_\Gamma-\bar{H}_\Gamma)\psi=(4N-4-M)\psi.
\]
Therefore,
\[
\left\|(H_\Gamma-\bar{H}_\Gamma) E^{\bar{H}_\Gamma}[0,M]\right\|\geq4N-4-M.
\]
The right-hand side diverges as $N\to\infty$. This contradicts the volume-independent bound \eqref{eq_AKL_original_2}. Hence an operator-norm formulation of this type cannot be volume-uniform in general.

\newpage
\section{Proof of Lemma \ref{lem_AKL2}}\label{S_AKL2}
\begin{lemma}\label{lem_infinite2}
We set $S_X$ as in Proposition \ref{prop_AKL}. Then, for any $s\in R(S_X)$ and any $\psi,\phi\in\mathcal{D}_{H_{X^c}}$, the following identity holds as a sesquilinear-form identity:
\[
\langle e^{\bar{s}H_{X^c}}\phi,H_X e^{-sH_{X^c}}\psi\rangle=\langle\phi,\exp(s\delta_{X^c})(H_X)\psi\rangle.
\]
\end{lemma}
\begin{proof}
The proof is the same analytic-continuation argument as in Lemma \ref{lem_infinite}.
\end{proof}

\begin{lemma}\label{lem_G1}
For any $s\in R(S_X)$, we have
\[
\left\|
e^{s\bar{H}_{X^c}}H_X e^{-s\bar{H}_{X^c}}
\right\|
\leq 2\|H_X\| +\frac{\|H_X\|}{1-|s|/S_X}.
\]
\end{lemma}
\begin{proof}
First, we prove the statement for the case $\operatorname{Re}(s)\geq0$. For simplicity, we write
\[
E_1:=E^{H_{X^c}}(-\infty,-M],\quad
E_2:=E^{H_{X^c}}(-M,M),\quad
E_3:=E^{H_{X^c}}[M,\infty).
\]
Then, we decompose
\begin{align*}
e^{s\bar{H}_{X^c}}H_X e^{-s\bar{H}_{X^c}}
&= e^{s\bar{H}_{X^c}} H_X e^{-s\bar{H}_{X^c}}E_3
+ E_1 e^{s\bar{H}_{X^c}} H_X e^{-s\bar{H}_{X^c}}(E_1+E_2)\\
&\qquad+(E_2+E_3) e^{s\bar{H}_{X^c}} H_X e^{-s\bar{H}_{X^c}}(E_1+E_2).
\end{align*}

For the first and second terms, we have
\begin{align*}
&\max\left\{ \left\|e^{-s\bar{H}_{X^c}} E_3\right\|,\,
\left\|E_1 e^{s\bar{H}_{X^c}}\right\| \right\}\leq e^{-\operatorname{Re}(s) M},\\
&\max\left\{\left\| e^{s\bar{H}_{X^c}}\right\|,\,
\left\|e^{-s\bar{H}_{X^c}}\right\|\right\} \leq e^{\operatorname{Re}(s) M}. 
\end{align*}
This yields the bounds
\[
\max\left\{
\left\|
e^{s\bar{H}_{X^c}} H_X e^{-s\bar{H}_{X^c}} E_3
\right\|,
\left\|
E_1 e^{s\bar{H}_{X^c}} H_X e^{-s\bar{H}_{X^c}}
\right\|
\right\}
\leq \|H_X\|.
\]

For the third term, for any $\phi\in\mathcal{D}_{H_{X^c}}$, we have
\begin{align*}
\left\|
e^{-sH_{X^c}}e^{s\bar{H}_{X^c}}
(E_2+E_3)\phi
\right\|^2
&\leq \|E_2\phi\|^2+
\int_M^\infty e^{2\operatorname{Re}(s)(M-\mu)} \|E^{H_{X^c}}(d\mu)\phi\|^2\\
&\leq \|E_2\phi\|^2+ \|E_3\phi\|^2
=\|(E_2+E_3)\phi\|^2.
\end{align*}
Similarly, we have
\[
\left\|
e^{sH_{X^c}}e^{-s\bar{H}_{X^c}}(E_1+E_2)\phi
\right\|
\leq \|(E_1+E_2)\phi\|.
\]
By Lemma \ref{lem_infinite2} and Lemma \ref{lem_derivation}(i), we obtain
\begin{align*}
&\left\|
(E_2+E_3) e^{s\bar{H}_{X^c}} H_X e^{-s\bar{H}_{X^c}} (E_1+E_2)\phi
\right\|\\
&\leq \|\exp(s\delta_{X^c})(H_X)\| \|\phi\|
\leq \frac{\|H_X\|}{1-|s|/S_X} \|\phi\|.
\end{align*}
Since $\phi\in\mathcal{D}_{H_{X^c}}$ is arbitrary and $\mathcal{D}_{H_{X^c}}\subset\mathcal{H}_\omega$ is dense, this yields
\[
\left\|
(E_2+E_3) e^{s\bar{H}_{X^c}} H_X e^{-s\bar{H}_{X^c}} (E_1+E_2)
\right\|
\leq \frac{\|H_X\|}{1-|s|/S_X}.
\]
Putting all these together, we obtain the desired estimate for $\operatorname{Re}(s)\geq0$. 

The case $\operatorname{Re}(s)\leq0$ follows in the same way by exchanging $E_1$ and $E_3$. 
\end{proof}

\begin{lemma}\label{lem_G}
For any $A\in\mathcal{B}(\mathcal{H}_\omega)$ and any $s\in R(S_X)$, we define
\[
G_1(s,A):=e^{s\bar{H}_\omega} e^{-s\bar{H}_{X^c}} A e^{s\bar{H}_{X^c}} e^{-s\bar{H}_\omega},
\]
\[
G_2(s,A):=e^{s\bar{H}_{X^c}} e^{-s\bar{H}_\omega} A e^{s\bar{H}_\omega}e^{-s\bar{H}_{X^c}}.
\]
Then we have 
\[
\max\{\|G_1(s,A)\|,\|G_2(s,A)\|\}
\leq \exp\left(2|s| \left( 2\|H_X\|+\frac{\|H_X\|}{1-|s|/S_X}\right)\right) \|A\|.
\]
\end{lemma}
\begin{proof}[Proof for Lemma \ref{lem_G}]
It is enough to show that
\begin{equation}\label{eq_trexp}
\left\|
e^{s\bar{H}_\omega}e^{-s\bar{H}_{X^c}}
\right\|
\leq \exp\left(2|s|\,\|H_X\|+\frac{|s|\,\|H_X\|}{1-|s|/S_X}
\right)
\end{equation}
in order to obtain the desired inequality. For any $t\in[0,1]$, we define
\[
K(t):=e^{ts\bar{H}_\omega}e^{-ts\bar{H}_{X^c}}.
\]
Then $K(t)$ is the unique solution of the $\mathcal{B}(\mathcal{H}_\omega)$-valued initial value problem
\[
\frac{d}{dt}K(t)=K(t) e^{ts\bar{H}_{X^c}} sH_X e^{-ts\bar{H}_{X^c}}
\quad\text{with}\quad K(0)=I. 
\]
This unique solution is given by the Dyson series and satisfies the following estimate (see e.g. \cite[Proposition 2.2]{Nachtergaele-Sims-Young}):
\[
\|K(t)\|\leq \exp\left(
|s| \sup_{t\in[0,1]}
\left\|
e^{ts\bar{H}_{X^c}} H_X e^{-ts\bar{H}_{X^c}}
\right\|\right).
\]

By Lemma \ref{lem_G1}, we have 
\[
\sup_{t\in[0,1]}
\left\|
e^{ts\bar{H}_{X^c}} H_X e^{-ts\bar{H}_{X^c}}
\right\|
\leq 2\|H_X\|+\frac{\|H_X\|}{1-|s|/S_X}
\]
and hence the bound \eqref{eq_trexp} follows.
\end{proof}

\begin{proof}[Proof of Lemma \ref{lem_AKL2}]
By Lemma \ref{lem_G}, the first inequality follows from
\begin{align*}
&\left\|E^{\bar{H}_\omega}[M,\infty)
\,A\,
E^{\bar{H}_\omega}(-\infty,N]\right\|\\
&\leq \left\|E^{\bar{H}_\omega}[M,\infty) e^{-\lambda \bar{H}_\omega}\right\|
\left\|G_1\left(\lambda,e^{\lambda\bar{H}_{X^c}}Ae^{-\lambda\bar{H}_{X^c}}\right)\right\|
\left\|e^{\lambda \bar{H}_\omega}E^{\bar{H}_\omega}(-\infty,N]
\right\|\\
&\leq\exp(-\lambda(M-N-8\|H_X\|)) \left\|e^{\lambda\bar{H}_{X^c}}Ae^{-\lambda\bar{H}_{X^c}}\right\|.
\end{align*}
Similarly, the second inequality can be obtained in the same way.
\end{proof}

\newpage
\section{A further illustration of the spectral-overlap method}\label{S_application}
To demonstrate the usefulness of the spectral-overlap method beyond the energy-truncation setting considered in the main text, we give an alternative proof of \cite[Lemma B.5]{Liu-Yi-Zhou-Zou2025} for a range-truncated Hamiltonian, replacing Weyl's inequality and resolvent analysis with a direct spectral-overlap argument.

\subsection{Notation and setup}
We consider a one-dimensional quantum spin system on $\Gamma=\mathbb{Z}$. We set local observables $\mathcal{A}_\Lambda$, the quasi-local algebra $\mathcal{A}_\mathrm{loc}$ and the observable algebra $\mathcal{A}$, in the same manner as Section \ref{S_notation} and \ref{S_infinite}.

In contrast to the main text, throughout this appendix we assume that the interaction $\Phi$ has algebraic two-point decay rather than finite range:

\begin{assumption}\label{ass_interaction2}
For some $\alpha>2$,
\[
J:=\sup_{x,y\in\mathbb{Z}} \sum_{\substack{Z\Subset\mathbb{Z}\\ x,y\in Z}} \frac{\|\Phi(Z)\|}{F(d(x,y))}<\infty,\quad F(d(x,y))=\frac{1}{(1+d(x,y))^\alpha}.
\]
\end{assumption}
Under Assumption \ref{ass_interaction2}, we also find the thermodynamic limit
\[
\alpha_t^\Gamma(A):=\lim_{\substack{\Lambda\Subset\Gamma\\ \Lambda\to\Gamma}}\alpha_t^\Lambda(A),\quad
H_\Lambda:=\sum_{X\subset\Lambda} \Phi(X),\quad
\alpha_t^\Lambda(A):=e^{itH_\Lambda}Ae^{-itH_\Lambda}
\]
exists for any $A\in\mathcal{A}_{\mathrm{loc}}$ and any $t\in\mathbb{R}$; see \cite[Theorem 3.5]{Nachtergaele-Sims-Young}.

\medskip
We fix a sufficiently large finite interval $I=[a,b]\cap\mathbb{Z}\Subset\Gamma$ as the subsystem of interest. We choose $\ell\in\mathbb{N}$ and $q\in\mathbb{N}$ so that $[a+(q+1)\ell,b-(q+1)\ell]\neq\emptyset$. We decompose the boundary region $\partial_{q\ell}I$ of width $q\ell$ into blocks of size $\ell$:
\begin{align*}
\partial_{q\ell} I
&:=[a-q\ell,a+q\ell]\cup[b-q\ell,b+q\ell]\\
&=\bigcup_{j=-q}^{q-1} [a+j\ell,a+(j+1)\ell]\cup 
\bigcup_{j=-q}^{q-1} [b+j\ell,b+(j+1)\ell].
\end{align*}

\begin{definition}[Range-truncated interaction]
The range truncated interaction $\tilde{\Phi}$ is defined from $\Phi$ by removing all interaction terms that connect one block to sites outside its $\ell$-neighborhood, that is, for any $j=-q,-q+1,\ldots q-1$, we set
\[
\mathcal{I}_{a,j}:=\{Z\Subset\mathbb{Z}\,;\,
Z\cap[a+j\ell,a+(j+1)\ell]\neq\emptyset
\ \text{and}\ 
Z\cap[a+(j-1)\ell,a+(j+2)\ell]^c\neq\emptyset\},
\]
\[
\mathcal{I}_{b,j}:=\{Z\Subset\mathbb{Z}\,;\,
Z\cap[b+j\ell,b+(j+1)\ell]\neq\emptyset
\ \text{and}\ 
Z\cap[b+(j-1)\ell,b+(j+2)\ell]^c\neq\emptyset\}.
\]
Then, the range truncated interaction is given by
\[
\tilde{\Phi}(Z):=
\begin{cases}
0 & Z\in\mathcal{I},\\
\Phi(Z) & Z\notin\mathcal{I},
\end{cases}
\quad\text{where}\quad
\mathcal{I}:=\bigcup_{j=-q}^{q-1} \Big(\mathcal{I}_{a,j}\cup\mathcal{I}_{b,j}\Big).
\]
\end{definition}

Under Assumption \ref{ass_interaction2}, the thermodynamic limit
\[
\tilde{\alpha}_t^\Gamma(A):=\lim_{\substack{\Lambda\Subset\Gamma\\ \Lambda\to\Gamma}} \tilde{\alpha}_t^\Lambda(A),\quad \tilde{H}_\Lambda:=\sum_{X\subset\Lambda}\tilde{\Phi}(X),\quad \tilde{\alpha}_t^\Lambda(A):=e^{it\tilde{H}_\Lambda}Ae^{-it\tilde{H}_\Lambda}
\]
exists for any $A\in\mathcal{A}_{\mathrm{loc}}$ and any $t\in\mathbb{R}$; see \cite[Theorem 3.5]{Nachtergaele-Sims-Young}.

We assume that $\omega$ is an equilibrium state for the dynamics $\alpha^\Gamma_t$, namely
\[
\omega(\alpha_t^\Gamma(A))=\omega(A)
\quad (A\in\mathcal{A},\ t\in\mathbb{R}).
\]
Write $(\mathcal{H}_\omega,\pi_\omega,\Omega)$ for the GNS triple associated with $\omega$. Then there exists a self-adjoint operator $H_\omega$ on $\mathcal{H}_\omega$ such that
\[
\pi_\omega(\alpha_t^\Gamma(A))=e^{itH_\omega}\pi_\omega(A)e^{-itH_\omega},\quad H_\omega\Omega=0.
\]
Moreover, we denote the discarded part of the interaction by
\[
\delta H:=\sum_{Z\in\mathcal{I}} \Phi(Z)\in\mathcal{A}.
\]
Here, by \cite[Lemma 2]{Kuwahara-Saito2021}, we have
\[
\|\delta H\|\leq\sum_{Z\in\mathcal{I}} \|\Phi(Z)\|
\leq 8qJ(1+\ell)^{-\alpha+2}<\infty.
\]
Since $\delta H$ is a bounded operator, we obtain
\[
\tilde{H}_\omega:=H_\omega-\pi_\omega(\delta H),\quad
\operatorname{Dom}(\tilde{H}_\omega)=\operatorname{Dom}(H_\omega),\quad
\pi_\omega(\tilde{\alpha}_t^\Gamma(A))=e^{it\tilde{H}_\omega}\pi_\omega(A)e^{-it\tilde{H}_\omega}.
\]
In the proof of this appendix, we write the bounded operator $\pi_\omega(\delta H)$ simply as $\delta H$. 

\subsection{Spectral Stability of range truncated Hamiltonian}
\begin{theorem}\label{thm_range}
If $\omega$ is a locally unique ground state with spectral gap $\Delta>2\|\delta H\|>0$, namely,
\[
\sigma(H_\omega)\subset \{0\}\cup [\Delta,\infty),\quad
\operatorname{Rank}(E^{H_\omega}(\{0\}))=1,
\]
where $E^{H_\omega}$ denotes the spectral projection of $H_\omega$. Then, $\tilde{H}_\omega$ still has a unique low-energy state with spectral gap, that is,
\[
\sigma(\tilde{H}_\omega)\subset [-\|\delta H\|,\|\delta H\|] \cup [\Delta-\|\delta H\|,\infty),\quad
\operatorname{Rank}(E^{\tilde{H}_\omega}[-\|\delta H\|,\|\delta H\|])=1
\]
and thus $\tilde{H}_\omega$ has a spectral gap $\tilde{\Delta}\geq\Delta-2\|\delta H\|$.

Moreover, let $\tilde{\Omega}\in\mathcal{H}_\omega$ be the vector spanning $E^{\tilde{H}_\omega}(-\infty,\|\delta H\|]$, then 
\[
\sqrt{1-|\langle\Omega,\tilde{\Omega}\rangle|^2}
\leq \frac{\|\delta H\|}{\Delta-\|\delta H\|}.
\]
In particular, if we set $\tilde{\Omega}$ as $\langle\Omega,\tilde{\Omega}\rangle\geq0$, then
\[
\|\Omega-\tilde{\Omega}\|
=\sqrt{2(1-|\langle\Omega,\tilde{\Omega}\rangle|)}
\leq \sqrt{2(1-|\langle\Omega,\tilde{\Omega}\rangle|^2)}
\leq \frac{\sqrt{2}\|\delta H\|}{\Delta-\|\delta H\|}.
\]
\end{theorem}

\begin{remark}
In the notation of the present paper, \cite[Lemma B.5]{Liu-Yi-Zhou-Zou2025} states that, under the assumption $\|\delta H\|<\Delta/4$, the range-truncated Hamiltonian has a unique ground state $\tilde{\Omega}$ and satisfies
\[
\tilde\Delta>\Delta-2\|\delta H\|
\quad\text{and}\quad
\|\Omega-\tilde{\Omega}\|
<\frac{\|\delta H\|}{\Delta-4\|\delta H\|}.
\]

Theorem \ref{thm_range} recovers the same type of spectral-gap stability conclusion by a direct spectral-overlap argument. This provides a further application of the spectral-overlap method and illustrates how the method can be used in this setting without relying on Weyl's inequality or resolvent estimates.
\end{remark}

\begin{proof}
For any $\psi\in\mathcal{H}_\omega$, we have
\begin{align*}
&\|H_\omega E^{H_\omega-\delta H}(-\Delta+\|\delta H\|,\Delta-\|\delta H\|)\psi\|\\
&\leq \|(H_\omega-\delta H)E^{H_\omega-\delta H}(-\Delta+\|\delta H\|,\Delta-\|\delta H\|)\psi\|\\
&\quad+\|\delta H\,E^{H_\omega-\delta H}(-\Delta+\|\delta H\|,\Delta-\|\delta H\|)\psi\|\\
&<\Delta.
\end{align*}
By Markov inequality,
\[
\|(I-E^{H_\omega}(-\Delta,\Delta))E^{H_\omega-\delta H}(-\Delta+\|\delta H\|,\Delta-\|\delta H\|)\psi\|<1.
\]
By Lemma \ref{lem_overlap},
\[
\operatorname{Rank}(E^{H_\omega-\delta H}(-\Delta+\|\delta H\|,\Delta-\|\delta H\|)) \leq\operatorname{Rank}(E^{H_\omega}(-\Delta,\Delta))=1.
\]

In the same manner, we have
\[
\|(I-E^{H_\omega-\delta H}[-\|\delta H\|,\|\delta H\|])E^{H_\omega}(\{0\})\|<1
\]
and
\[
1=\operatorname{Rank}(E^{H_\omega}(\{0\}))\leq
\operatorname{Rank}(E^{H_\omega-\delta H}[-\|\delta H\|,\|\delta H\|]).
\]
Consequently, we have
\[
\operatorname{Rank}(E^{H_\omega-\delta H}[-\|\delta H\|,\|\delta H\|])=\operatorname{Rank}(E^{H_\omega-\delta H}(-\Delta+\|\delta H\|,\Delta-\|\delta H\|))=1.
\]

We set $\tilde{\Omega}\in\mathcal{H}_\omega$ so that 
\[
|\tilde{\Omega}\rangle\langle\tilde{\Omega}|
=E^{H_\omega-\delta H}(-\Delta+\|\delta H\|,\Delta-\|\delta H\|).
\]
Then we have
\[
\|(I-E^{H_\omega-\delta H}(-\Delta+\|\delta H\|,\Delta-\|\delta H\|))E^{H_\omega}(\{0\})\|\leq\frac{\|\delta H\|}{\Delta-\|\delta H\|}.
\]
By Lemma \ref{lem_overlap},
\[
|\langle\Omega,\tilde{\Omega}\rangle|=
\|E^{H_\omega-\delta H}(-\Delta+\|\delta H\|,\Delta-\|\delta H\|)\Omega\|
\geq \sqrt{1-\left(\frac{\|\delta H\|}{\Delta-\|\delta H\|}\right)^2}.
\]
\end{proof}
\end{document}